\documentclass[conference,10pt]{IEEEtran}
\IEEEoverridecommandlockouts
\usepackage[utf8]{inputenc}
\usepackage{amsmath,epsfig} 
\usepackage{amssymb,amsthm}

\usepackage{url}
\usepackage{epstopdf}
\usepackage{multirow}
\usepackage{setspace}
\usepackage[noadjust]{cite}
\usepackage{xspace}
\usepackage{tabularx}
\usepackage{changepage}
\usepackage{color}
\usepackage{graphicx,bm}
\usepackage{color,cite}
\usepackage{times}
\usepackage{enumerate,type1cm}
\usepackage{amsfonts,relsize}
\usepackage{relsize}
\usepackage{fancybox}
\usepackage[english]{babel}
\usepackage{fancybox}
\usepackage{bbm}
\usepackage{tikz}
\usetikzlibrary{fit,positioning}
\usetikzlibrary{bayesnet}
\usepackage{orcidlink}
\usepackage{caption}

\def\BibTeX{{\rm B\kern-.05em{\sc i\kern-.025em b}\kern-.08em
		T\kern-.1667em\lower.7ex\hbox{E}\kern-.125emX}}

\theoremstyle{definition}
\newtheorem{thm}{\textbf{Theorem}}

\long\def\symbolfootnote[#1]#2{\begingroup%
	\def\thefootnote{\fnsymbol{footnote}}\footnote[#1]{#2}\endgroup}

\newcommand{\beq}{\begin{equation}}
\newcommand{\eeq}{\end{equation}}
\newcommand{\beqa}{\begin{eqnarray}}
\newcommand{\eeqa}{\end{eqnarray}}
\usepackage{float}
\usepackage{colortbl}

\tikzset{
	startstop/.style={
		rectangle, 
		rounded corners,
		minimum width=3cm, 
		minimum height=0.5cm,
		align=center, 
		draw=black, 
	},
	process/.style={
		rectangle, 
		minimum width=3cm, 
		minimum height=0.5cm, 
		align=center, 
		draw=black, 
	},
	decision/.style={
		rectangle, 
		minimum width=3cm, 
		minimum height=0.5cm, align=center, 
		draw=black, 
	},
	arrow/.style={thick,->,>=stealth},
	dec/.style={
		ellipse, 
		align=center, 
		draw=black, 
	},
}


\makeatletter
\def\BState{\State\hskip-\ALG@thistlm}
\makeatother

	{\end{adjustwidth}}

\newcommand{\widesim}[2][1.5]{
  \mathrel{\overset{#2}{\scalebox{#1}[1]{$\sim$}}}
}

\usepackage{array}

\makeatletter
\newcommand{\removelatexerror}{\let\@latex@error\@gobble}
\makeatother

\allowdisplaybreaks

\usepackage{cite}
\usepackage{amsmath,amssymb,amsfonts,amsthm}
\usepackage{algorithmic}
\usepackage{graphicx}
\usepackage{textcomp}
\usepackage{xcolor}
\usepackage{graphicx}
 \usepackage{marginnote}
 \usepackage{fancyhdr}
 \usepackage{epstopdf}
\usepackage{anyfontsize}
\usepackage{makeidx}
\usepackage[utf8]{inputenc}
\usepackage{multirow}
\usepackage{multicol}
\usepackage[none]{hyphenat}
\usepackage{lastpage}
\usepackage[T1]{fontenc}
\usepackage{url} 
\usepackage{epstopdf}
\usepackage{booktabs}
\usepackage{multicol}
\usepackage{fancyhdr}
\usepackage{indentfirst}
\usepackage{color}
\usepackage[section]{placeins}
\usepackage{float}
\usepackage[font=small]{caption}
\usepackage{subcaption}
\usepackage{amsmath, amsthm, amssymb}
\usepackage{bm,bbm}
\usepackage{commath}
\usepackage[ruled,vlined,linesnumbered]{algorithm2e}
\usepackage{stfloats}
\usepackage{afterpage}
\usepackage{etoolbox}
\newcommand{\RNum}[1]{\uppercase\expandafter{\romannumeral #1\relax}}
\def\BibTeX{{\rm B\kern-.05em{\sc i\kern-.025em b}\kern-.08em
    T\kern-.1667em\lower.7ex\hbox{E}\kern-.125emX}}

\title{Channel State Information Based User Censoring in Irregular Repetition Slotted Aloha \vspace{-2.5mm}
}

\author{\IEEEauthorblockN{Chirag Ramesh Srivatsa, \IEEEmembership{Graduate Student Member, IEEE,} and Chandra R. Murthy, \IEEEmembership{Fellow, IEEE}}
\IEEEauthorblockA{{Dept. of CPS and Dept. of ECE,  Indian Institute of Science, Bengaluru, India (e-mail: \{chiragramesh, cmurthy\}@iisc.ac.in).}}  
}

\begin{document}
\maketitle
\begin{abstract}
Irregular repetition slotted aloha (IRSA) is a massive random access protocol which can be used to serve a large number of users while achieving a packet loss rate (PLR) close to zero.
However, if the number of users is too high, then the system is interference limited and the PLR is close to one.
In this paper, we propose a variant of IRSA in the interference limited regime, namely Censored-IRSA (C-IRSA), wherein users with poor channel states censor themselves from transmitting their packets.
We theoretically analyze the throughput performance of C-IRSA via density evolution.
Using this, we derive closed-form expressions for the optimal choice of the censor threshold which maximizes the throughput while achieving zero PLR among uncensored users.
Through extensive numerical simulations, we show that C-IRSA can achieve a $4\times$ improvement in the peak throughput compared to conventional IRSA. 
\end{abstract}

\vspace{-0.2cm}
\section{Introduction} \label{sec_intro}
Massive machine-type communications (mMTC) is an evolving use-case, expected to serve about a million users per square km~\cite{ref_chen_jsac_2021}.
In this context, irregular repetition slotted aloha (IRSA)~\cite{ref_liva_toc_2011} is a distributed massive random access protocol which has received much attention in the literature~\cite{ref_srivatsa_chest_tsp_2022,ref_srivatsa_uad_tsp_2022}.
The performance of IRSA mainly depends on the system load, i.e., the number of users participating in the protocol per frame. 
At low system loads,  the system is not interference-limited, and the packet loss rate (PLR) is close to zero.
As the system load increases beyond a so-called \emph{inflection load}, IRSA becomes interference limited, and the system throughput rapidly drops to zero, with PLR approaching one~\cite{ref_srivatsa_chest_tsp_2022}.
In this paper, we address the issue of the poor throughput of IRSA in the high load regime by introducing a \emph{distributed self-censoring scheme,} which allows the system to maintain the throughput at the maximum possible value even as the load increases.

In IRSA, users transmit packet replicas in multiple randomly selected slots (within a frame) to a base station (BS); the latter decodes the packets using successive interference cancellation (SIC)~\cite{ref_paolini_tit_2015}. 
If the BS successfully decodes a user in a slot, it uses the decoded data to perform SIC in all other slots in which that user has transmitted a replica.
The decodability of any user in IRSA depends on the signal to interference plus noise ratio (SINR) of that user~\cite{ref_khaleghi_pimrc_2017}.
If the users have poor channel states or there are many collisions in a slot resulting in increased multi-user interference (MUI), then the SINR decreases, 
leading to users not getting decoded.

When the system load is low, all users with sufficiently good channel states can be decoded at the end of the SIC process~\cite{ref_clazzer_icc_2017}.
However, as the load increases, more collisions lead to the system becoming interference limited, and the throughput quickly drops close to zero~\cite{ref_paolini_tit_2015}.
To tackle this issue, in this paper, we propose a modified IRSA protocol, as follows. 
The BS transmits a pilot signal at the start of each frame, using which the users estimate their channel state information (CSI). 
Then, users with poor CSI  \emph{self-censor}, i.e., they refrain from transmitting, thereby reducing collisions and enable the decoding of the transmissions from the active users to succeed.
Here, choosing too high a (CSI-based) censor threshold leads to very few users participating in the IRSA protocol, while too low a threshold leads to too many collisions and the system becoming interference-limited. 
Thus, given the system load, there is an optimal threshold that maximizes the throughput. 
The censor threshold is calculated by the BS based on the load and the SNR, and its value is periodically broadcast to the users.
Note that this approach retains the fully distributed nature of IRSA.

IRSA has been studied for the collision channel~\cite{ref_liva_toc_2011}, with multiple antennas~\cite{ref_srivatsa_chest_tsp_2022},  with activity detection~\cite{ref_srivatsa_uad_tsp_2022},  with path loss~\cite{ref_khaleghi_pimrc_2017}, for the Rayleigh fading channel~\cite{ref_clazzer_icc_2017}, with channel estimation errors~\cite{ref_srivatsa_spawc_2019},  and with multi-cell effects~\cite{ref_srivatsa_spawc_2022}.
Density evolution has been used to characterize the asymptotic throughput of IRSA~\cite{ref_srivatsa_chest_tsp_2022,
ref_clazzer_icc_2017,ref_khaleghi_pimrc_2017}.
Variants of aloha such as $K$-repetition have also been  studied~\cite{ref_choi_wcl_2021, ref_ding_cscn_2021}.
However, none of these papers address the dramatic increase in PLR and the corresponding reduction in throughput as the system load increases, which is our focus in this paper. 
Our specific contributions are as follows: 
\begin{enumerate}
\item We propose a censored-IRSA (C-IRSA) protocol in the interference limited regime, where users with poor CSI self-censor to decrease the effective system load, thereby enabling the uncensored users to be decoded at the BS.
\item We theoretically analyze the asymptotic performance of C-IRSA using density evolution.
\item We provide the optimal choice of the censor threshold with which the PLR of uncensored users can be driven close to zero at all system loads, while maintaining the throughput of the system at its highest value. 
\end{enumerate}

With CSI-based censoring in C-IRSA, we can achieve a $4\times$ improvement in the peak performance of the system compared to conventional IRSA.
Further, C-IRSA can be operated at the peak performance for all system loads, whereas the throughput of conventional IRSA becomes zero at high loads.

\textit{Notation:} The symbols $a$, $ \mathbf{a}$, and $\mathbf{A}$, denote a scalar, a vector, and a matrix, respectively. 
$[N]$ denotes the set $\{1,2,\ldots,N\}$.
$\mathbbm{1}\{ \cdot \}$, $|\cdot|$, $ [\cdot]^*$, and $\mathop{{}\mathbb{E}}[\cdot]$,  denote the indicator, magnitude (or cardinality of a set),  conjugate, and expectation, respectively.

\vspace{-0.2cm}
\section{System Model} \label{sec_sys_model}
We consider an IRSA system with $M$ single-antenna users communicating with a central BS over a frame consisting of $T$ slots.
The BS is located at the cell center, and the users are arbitrarily located within the cell.
mMTC applications use similar settings as narrowband internet of things, which uses a narrow bandwidth of $180$~kHz~\cite{ref_chen_jsac_2021}. 
Over this band, the channel can be assumed to be flat and Rayleigh block fading.
The BS allocates a pre-specified band to all the users in the system and the $M$ users transmit their packets within this band. 
The \emph{system load,} $L$, is defined as the ratio of the number of users to the number of slots, $L\! \triangleq \!M/T$. 
In conventional IRSA,  in each frame, the users randomly select a subset of the slots, and transmit replicas of their packets in the chosen slots.
The access of the $T$ slots in a given frame can be represented as a \emph{binary access pattern matrix} $\mathbf{G} \!\in\! \{0,1\}^{T \times M}$~\cite{ref_srivatsa_chest_tsp_2022}.
The $(t,m)$th element of $\mathbf{G}$, denoted by $g_{tm}$, equals $1$ if the $m$th user transmits its packet in the $t$th slot, and $g_{tm} \!=\! 0$ otherwise.
In such a protocol, when $L$ is high, there will be too many collisions in any slot, leading to a failure of the SIC-based decoding process (described below), resulting in low throughput.

The C-IRSA protocol we propose works as follows. 
At the start of each frame, the BS transmits a pilot signal, using which the users estimate their channel state. 
The users participate in the IRSA protocol if and only if the magnitude squared of their channel exceeds a censor threshold denoted by $\nu$. 
We refer to the users who self-censor  as \emph{inactive} or \emph{censored} users, and the other users as \emph{active} or \emph{uncensored} users. 
A censored user can sleep till the next time it has data to transmit, by when its channel state would change. 
At the BS, the active users' packets are decoded using the SIC process as with the conventional IRSA protocol.

The $m$th user transmits a symbol $x_m$ with $\mathbb{E}[x_m] \!=\! 0$ and $\mathbb{E}[|x_m|^2] \!=\! 1$. 
The received signal $y_{t} $ at the BS in the $t$th slot~is
\begin{align}
y_{t}  &=  \textstyle{\sum\nolimits_{m=1}^M} \sqrt{\rho_0}  a_{m} g_{tm} h_{m} x_m  + n_{t}, \label{eqn_init_rx_data}
\end{align}
where $h_{m} \!\widesim[1.5]{\text{i.i.d.}}\! \mathcal{CN} (0, 1) \ \forall m \in [M]$ is the uplink fading channel of the $m$th user, \textcolor{black}{assumed independent across users and frames};  $a_{m} \!=\! \mathbbm{1}\{ |h_{m}|^2 \!\geq\! \nu \} $ is the activity coefficient of the $m$th user,  and $n_{t}\! \widesim[1.5]{\text{i.i.d.}}\! \mathcal{CN}(0,1)$ is the complex AWGN at the BS. Also, $\rho_0 \!\triangleq \!P \sigma_{\tt{h}}^2/N_0$ denotes the signal to noise ratio (SNR) of any user (in the absence of any collisions), where $P$ is the transmit power, $\sigma_{\tt{h}}^2$ is the fading variance, and $N_0$ is the noise variance. 
Here, we assume that the users perform path loss inversion based power control to ensure the same average received power levels of all users, which in turn ensures fairness.
In addition, inverting only the path loss rather than full channel inversion helps with \emph{capture effect}, which allows some of the users to be decoded in slots where there are collisions,  improving the throughput~\cite{ref_srivatsa_chest_tsp_2022}.
We denote the set of active users by $\mathcal{A}\! \triangleq\! \{ i \in [M] | |h_i|^2 \!\geq\! \nu \}$, the number of active users by $M_a\! \triangleq\! |\mathcal{A} |$, and the \emph{active load} $L_a$ by~$L_a \!\triangleq \! M_a/T $.

\subsubsection{SIC-based Decoding}\label{sec_sic_dec}
The BS iteratively processes the received signal.
In each slot, the BS attempts to decode the users' packets.
If a user is successfully decoded, which can be verified via a cyclic redundancy check, then using the decoded data, the BS performs SIC in all slots in which that user has transmitted a packet~\cite{ref_liva_toc_2011}.\footnote{The set of slots where the user's packet is repeated can be included in the header of the packet.}
This process repeats and the decoding at the BS proceeds in iterations.

We use the SINR threshold model to abstract the decodability of any packet: a packet can be decoded correctly if and only if its SINR is above a threshold $\gamma_{\text{th}}\! \ge\! 1$~\cite{ref_clazzer_icc_2017,ref_srivatsa_chest_tsp_2022}.
To evaluate the performance of C-IRSA with the SINR threshold model, we first compute the SINRs achieved by all the users in all the slots in any decoding iteration.
If there is a user with SINR $\ge\! \gamma_{\text{th}}$ in some slot, we consider that packet as successfully decoded and remove the contribution of that user's packet from all other slots in which that user has transmitted a replica~\cite{ref_khaleghi_pimrc_2017}.
We then proceed to the next decoding iteration and recompute the SINRs for all users yet to be decoded.
This process stops when no additional users are decoded in two successive iterations.
The throughput $\mathcal{T}$ is calculated as the number of correctly decoded unique packets divided by the number of slots.

\begin{algorithm}[t]
\DontPrintSemicolon
\SetAlgoLined
\SetKwInOut{Input}{Input}
\Input{$ T, M, \rho_0, \mathbf{G}, k_{\max},\mathcal{A} = \{ i \in [M]||h_{i}|^2 \geq \nu \}$}
\textbf{Initialize:} \thinspace $ \mathcal{S}_1 = [M] $ \\
\For{$k = 1,2, \ldots, k_{\max}$}{
\For{$t = 1,2,\ldots, T$}{
Evaluate the SINR $\rho_{ti}^k, \ \forall i \in \mathcal{S}_k $ from \eqref{eqn_sinr_init} \\
If $\rho_{ti}^k \geq \gamma_{\text{th}}$,  remove user $i$ from $ \mathcal{S}_k$ and perform SIC in all slots where $g_{ti} = 1$  \\
}
}
\textbf{Output:}  $\mathsf{PLR} = |\mathcal{S}_{k_{\max}}|/M, \ \mathcal{T} = M(1 -\mathsf{PLR})/T,$ $\mathsf{PLR}_a = |\mathcal{A} \cap \mathcal{S}_{k_{\max}}|/|\mathcal{A}|.$
\caption{Performance Evaluation of C-IRSA}
\label{algo_perf_eval}
\end{algorithm}

The calculation of the SINR of the users is as follows.
We define $\mathcal{S}_k$ as the set of users not decoded upto the $k$th iteration with $\mathcal{S}_k^m \triangleq \mathcal{S}_k \setminus \{m\}$ and $\mathcal{S}_1 = [M]$.
We can write the received signal in the $t$th slot in the $k$th decoding iteration as
\begin{align}
y_{t}^k  &=  \textstyle{\sum\nolimits_{i\in \mathcal{S}_k}} \sqrt{\rho_0}  a_{i} g_{ti} h_{i} x_i  + n_{t}. \label{eqn_rx_data}
\end{align}
In order to decode the $m$th user, we first compute the processed signal $\tilde{y}_{tm}^k \triangleq h_{m}^* y_{t}^k$, which can be written as
\begin{align*}
\tilde{y}_{tm}^k \! = \! \sqrt{\rho_0} a_{m} g_{tm} |h_{m}|^2 x_m \! \! + \! \!  \textstyle{\sum\nolimits_{i\in \mathcal{S}_k^m}} \sqrt{\rho_0}  a_{i} g_{ti} h_{m}^*h_{i} x_i  \! + \! h_{m}^*n_{t}, \label{eqn_post_comb}
\end{align*}
where the first term $T_1 \triangleq \sqrt{\rho_0} a_{m} g_{tm} |h_{m}|^2 x_m$ is the desired signal, the second term $T_2 \triangleq \textstyle{\sum\nolimits_{i\in \mathcal{S}_k^m}} \sqrt{\rho_0}  a_{i} g_{ti} h_{m}^*h_{i} x_i$ is the MUI, and $T_3 \triangleq h_{m}^*n_{t}$ is noise.
Since noise is uncorrelated with the other terms and the data streams of distinct users are uncorrelated, the terms $T_1, T_2, $ and $T_3$ are all uncorrelated with each other. 
The power in the received signal is a sum of the powers of the terms.
Thus,  the SINR, $\rho_{tm}^k$, of the $m$th user in the $t$th slot in the $k$th iteration, can be computed as
\begin{equation}
\rho_{tm}^k = \dfrac{\rho_0 a_{m} g_{tm} | h_{m} |^2}{  1  +  \sum\nolimits_{i \in \mathcal{S}_k^m} \rho_0 a_{i} g_{ti} | h_{i} |^2  }. \label{eqn_sinr_init} 
\end{equation}
The performance of C-IRSA can now be computed as detailed in Alg.~\ref{algo_perf_eval}.
Here, the decoding proceeds for $k_{\max}$ iterations, and the output is the system throughput, $\mathcal{T}$, the packet loss rate (PLR) of the active users, $\mathsf{PLR}_a$, and the system PLR, $\mathsf{PLR}$.

\emph{Remarks:} The threshold $\nu$ can be declared by the BS during pilot transmission based on the  system load; the optimal choice of $\nu$ is discussed in the sequel. 
Also, we ignore channel estimation errors in determining whether the user remains active/inactive and in calculating the SINR in \eqref{eqn_sinr_init}. 
However, it is straightforward to include these effects using the results in \cite{ref_srivatsa_chest_tsp_2022}. 
Finally, the BS can determine which users are active in each frame, for example, using the user activity detection (UAD) algorithm presented in~\cite{ref_srivatsa_uad_tsp_2022}. It is shown in \cite{ref_srivatsa_uad_tsp_2022} that a short pilot transmission from the users for channel estimation at the BS is also sufficient for accurate UAD.

\section{Theoretical Analysis of C-IRSA} \label{sec_de}
In the previous section, we described an \emph{empirical} approach to evaluate the performance of C-IRSA, given by Alg.~\ref{algo_perf_eval}.
We now characterize the {theoretical} performance of C-IRSA using density evolution (DE)~\cite{ref_liva_toc_2011,ref_srivatsa_chest_tsp_2022}.
SIC-based decoding can be viewed as message passing on a bipartite graph \cite{ref_khaleghi_pimrc_2017}, and thus C-IRSA can be decoded on graphs. 
The bipartite graph is made up of the user nodes on one side, the slot nodes on the other side, and the edges between them.
An edge connects a user node to a slot node if and only if that user has transmitted a packet in that corresponding slot.
DE is applicable as ${M_a}$ and $T \rightarrow \infty $ with a fixed $L_a={M_a}/T$~\cite{ref_clazzer_icc_2017}.
Hence, we describe the DE process in terms of only the active load $L_a$.
Due to lack of space, we only outline the high-level steps in the analysis here. Detailed discussion of the DE technique can be found in several references~\cite{ref_srivatsa_chest_tsp_2022,ref_liva_toc_2011,
ref_clazzer_icc_2017,ref_khaleghi_pimrc_2017}.

The repetition factor of a user is the number of replicas the user has transmitted in a given frame, whereas the collision factor of a slot is the number of packets that have collided in that slot.
The \emph{node-perspective user degree distribution} is the set of probabilities $ \{\phi_d\}_{d=2}^{d_{\max}} $, where $\phi_d$ is the probability that a user has a repetition factor $d$; with minimum and maximum repetition factors of $2$ and $d_{\max}$, respectively.
The \emph{edge-perspective user degree distribution} is the set of probabilities $ \{\lambda_d\}_{d=2}^{d_{\max}} $, where $\lambda_d = d\phi_d / \phi'(1)$ is the probability that an edge is connected to a user with repetition factor $d$.
The corresponding polynomial representations of the node- and edge- perspective user degree distributions are
\begin{align}
\phi (x) = \textstyle{\sum\nolimits_{d = 2}^{d_{\max}}} \phi_d x^d, \ \
\lambda (x) = \sum_{d=2}^{d_{\max}} \lambda_d x^{d-1}, 
\end{align}
respectively.
The average repetition factor is $\bar{d}  \triangleq \sum_d d \phi_d$.

The degree distributions defined above are now used to find a pair of interdependent \emph{failure probabilities} denoted by ``$p_i$'' and ``$q_i$'' in the $i$th decoding iteration.
The user and slot nodes exchange failure messages along an edge when a decoding failure happens, i.e., when that user has not been decoded in that slot in the current decoding iteration.
The probability that an edge carries a failure message from a slot node to a user node is denoted by~$p_i$, whereas the probability that an edge carries a failure message from a user node to a slot node is denoted by~$q_i$.
Using the edge-perspective user degree distribution, the failure probability $q_i$ is calculated as
\begin{align}
q_i = \textstyle{\sum\nolimits_{d=2}^{d_{\max}}} \lambda_d p_{i-1}^{d-1} = \lambda (p_{i-1}). \label{eqn_de_qi}
\end{align}
Here,  the probability that an edge carries a failure message in the $i$th iteration given that it is connected to a user node with repetition factor $d$ is $p_{i-1}^{d-1}$. 
If all the other $d-1$ incoming edges to that user node carry failure messages in the previous iteration, then the edge will carry a failure message from that user node in the $i$th iteration.
The failure probability $p_i$ is calculated as in~\cite{ref_clazzer_icc_2017,ref_srivatsa_chest_tsp_2022} as
\begin{align}
p_i &= 1 - e^{-L_a\bar{d} q_i} \textstyle{\sum\nolimits_{r=1}^\infty} \theta_r ({L_a\bar{d} q_i})^{r-1}/(r-1)! \triangleq f(q_i) . \label{eqn_de_pi}
\end{align}
Here, $\theta_r$ is the probability that a reference packet gets decoded in any iteration in a slot of degree $r$ using only intra-slot SIC~\cite{ref_clazzer_icc_2017}.
Intra-slot SIC refers to interference cancellation within the same slot a user is decoded in, whereas inter-slot SIC refers to interference cancellation in a different slot.
We now describe the evaluation of $\theta_r$, which is the crucial step in computing the throughput.

\begin{thm} \label{thm_thetar}
For the Rayleigh block-fading channel with an SNR of $\rho_0$, a censor threshold $\nu$, and a decoding threshold $\gamma_{{\text{\rm{th}}}}$, the probability that a reference packet gets decoded in a slot of degree $r$ using only intra-slot SIC, can be obtained as
\begin{align}
\theta_{r} &= {\sum\limits_{k=1}^r} \dfrac{ \exp ( r \nu - (r-k) \nu \bar{\gamma}_{{\text{\rm{th}}}, k} -\rho_0^{-1} (\bar{\gamma}_{{\text{\rm{th}}}, k} - 1)  )}{ r \ \bar{\gamma}_{{\text{\rm{th}}}, k}^{r - (k+1)/2}},
\end{align}
where $\bar{\gamma}_{{\text{\rm{th}}}, k} = (1+ \gamma_{{\text{\rm{th}}}})^k$,  and $\nu\leq \rho_0^{-1} \gamma_{{\text{\rm{th}}}}$.
\end{thm}
\begin{proof}
See Appendix \ref{appendix_theta_r}.
\end{proof}

\noindent \emph{Remark:} When $\nu \! = \! 0$, i.e., there is no censoring, the expression for $\theta_r$ matches with the results by Clazzer et al.~\cite{ref_clazzer_icc_2017}.

In DE, $q_i \!=\! \lambda(p_{i-1})$ and $p_i \!=\! f(q_i)$ are calculated recursively as functions of each other using \eqref{eqn_de_qi} and \eqref{eqn_de_pi}, with either $q_0 \!=\! 1$ or $p_0 \!=\! f(1)$.
At the end of decoding, the failure probability is $p_{\infty} \!=\! \lim_{i \rightarrow \infty} p_i $.
The probability that a packet from a user with repetition factor $d$ does not get decoded at all is $(p_{\infty})^d$. 
Therefore, the asymptotic packet loss rate of the active users ($\mathsf{PLR}_a$), which is the fraction of packets of active users that are not decoded at the BS, is calculated as
\begin{align} 
\mathsf{PLR}_a &=  \phi(p_{\infty}) = \textstyle{\sum\nolimits_{d=2}^{{d_{\max}}}} \phi_d (p_{\infty})^d.
\end{align}
We denote the cumulative distribution function (CDF) and the complementary CDF of the exponential distribution (of $|h|^2\! \sim\! \exp(1)$) evaluated at $x$ by ${\rm F}(x)$ and $\bar{\rm F}(x) \!\triangleq \!1 \!- \!{\rm F}(x)$, respectively.
The active load $L_a$ of the system is $L_a \!=\! L\bar{\rm F}(\nu)$. 
Since the fraction of censored users is ${\rm F}(\nu)$, the effective $\mathsf{PLR}$ of the system (including censored users) can be calculated as 
\begin{align}
\mathsf{PLR} = {\rm F}(\nu) + \bar{\rm F}(\nu)\mathsf{PLR}_a. \label{eqn_plr_plra}
\end{align}
The throughput $\mathcal{T}$ of the users in the system can now be obtained from the asymptotic $\mathsf{PLR}$ as
\begin{align}
\mathcal{T} = L (1 - {\mathsf{PLR}}) = L_a (1 - {\mathsf{PLR}}_a).  \label{eqn_plr_thpt}
\end{align}
The iterations $p_{i} \! = \! f(\lambda(p_{i-1}))$ converge to $p_\infty \!=\! 0$ if the active load $L_a \!< \!L_a^*$, asymptotically~\cite{ref_srivatsa_chest_tsp_2022,ref_liva_toc_2011}.
Here, $L_a^*$ is called the \emph{active inflection load} of the system, and it corresponds to a \emph{system inflection load} of $L^*  \! = \! L_a^*/\bar{\rm F}(\nu)$, with a threshold $\nu$.
For $L_a  \! < \! L_a^*$,  since $p_\infty \!=\! 0$, we have $\mathsf{PLR}_a \!=\! 0$, $\mathsf{PLR}  \! = \! {\rm F}(\nu)$, and $\mathcal{T}  \! = \! L\bar{\rm F}(\nu) \! = \! L_a$. 
For any $L_a \! \geq \! L_a^*$,  $\mathsf{PLR}_a$ does not converge to $0$, and $\mathcal{T}$ decreases monotonically with $L_a$.
Also, from \eqref{eqn_plr_plra}, we see that $\mathsf{PLR} \! \geq  \! {\rm F}(\nu)$, and thus,  $\mathcal{T} \!\leq\! L\bar{\rm F}(\nu)$.

\subsubsection{Choice of Threshold}\label{sec_choice_nu}
In order to choose $\nu$, we first choose a target PLR for the active users, ${\sf{PLR}}_{a, \text{tgt}}$, which is a maximum permissible PLR among the active users.
Let $L_{\text{tgt}} $ be the \emph{target load},  which is the minimum $L$ at which the system achieves an active PLR of ${\sf{PLR}}_{a, \text{tgt}}$, with $\nu \!=\! \rho_0^{-1}\gamma_{\text{th}}$.
At $L_{\text{tgt}} $, the active load is $L_a \!=\! L_{\text{tgt}}\bar{\rm F}( \rho_0^{-1}\gamma_{\text{th}})$, with a corresponding throughput of ${\mathcal{T}}_{\text{tgt}}$.
For a load $L \!\geq\! L_{\text{tgt}}$,  we wish to continue to operate at the same PLR of ${\sf{PLR}}_{a, \text{tgt}}$, to keep the throughput fixed at ${\mathcal{T}}_{\text{tgt}}$.
This can be done by ensuring the same active load $L_a$ at $L$ and $L_{\text{tgt}}$.
Thus, we need to choose $\nu$ such that 
\begin{align}
L_a = L\bar{\rm F}(\nu)  = L_{\text{tgt}}\bar{\rm F}( \rho_0^{-1}\gamma_{\text{th}}).
\end{align}
Since $\bar{\rm F}(x) \!=\! \exp(-x)$, we obtain
\begin{align}
\nu = \log(L/L_a) = \log (L/L_{\text{tgt}}) + \rho_0^{-1}\gamma_{\text{th}}.
\end{align}
The above is valid when $L\!\geq\!L_a$ or $L \!\geq\! L_{\text{tgt}}$.
When $L \!<\! L_{\text{tgt}}$, 
as we will see in Fig.~\ref{fig_userc_thpt_vs_activeL}, the threshold that maximizes the throughput
occurs at $\nu \!=\! \rho_0^{-1}\gamma_{\text{th}}$.
An intuitive reason for this is that the probability of decoding a user, if that user was the only one transmitting in a slot, is $\theta_1 \! =  \! \text{Pr} ( |h_{1}|^2 \! \geq \! \rho_0^{-1} \gamma_{\text{th}} \ | \ |h_{1}|^2 \! \geq \! \nu)$  $=\! \exp (\nu \!- \!\rho_0^{-1}\gamma_{\text{th}}) \cdot \mathbbm{1} \{ \nu \!\leq\! \rho_0^{-1}\gamma_{\text{th}} \} + \mathbbm{1} \{ \nu \!>\! \rho_0^{-1}\gamma_{\text{th}} \} $, when the threshold is $\nu$.
So if we set $\nu \!>\! \rho_0^{-1}\gamma_{\text{th}}$ or $\nu \!<\! \rho_0^{-1}\gamma_{\text{th}}$, we are censoring more or fewer users than required, respectively.
Thus, the optimal choice of the censor threshold is given by the function $ g(\cdot, \cdot)$ defined as
\begin{align}
\! \nu = g(L,L_{\text{tgt}}) & \triangleq 
\begin{cases}
\rho_0^{-1}\gamma_{\text{th}}, & \! L < L_{\text{tgt}}, \\
\log (L/L_{\text{tgt}}) + \rho_0^{-1}\gamma_{\text{th}}, & \! L \geq L_{\text{tgt}}.
\end{cases}
\end{align}
For $\nu \!=\! \rho_0^{-1}\gamma_{\text{th}}$, the system inflection load is  $L^* \!=\! L_a^*/\bar{\rm F}( \rho_0^{-1}\gamma_{\text{th}})$.
For $L_{\text{tgt}} \!<\! L^*$, the set of functions $\{ g(\cdot,\cdot) \}$ achieve ${\sf{PLR}}_a\! \leq\! {\sf{PLR}}_{a,\text{tgt}}$ among the set of active users.
In practice, we set a low target PLR of ${\sf{PLR}}_{a,\text{tgt}} \!\approx\! 10^{-3}$ or $10^{-4}$.

\section{Numerical Results} \label{sec_results}
In this section, we illustrate the performance of C-IRSA via Monte Carlo simulations. We also show how C-IRSA helps to overcome packet losses both due to poor CSI as well as due to MUI, as the load increases. 
In every simulation, we generate independent realizations of the channels and the access pattern matrix, and empirically evaluate the throughput of C-IRSA using Alg.~\ref{algo_perf_eval}.
We also evaluate the theoretical  throughput of C-IRSA as discussed in Sec.~\ref{sec_de} and provide insights into the impact of various system parameters on the performance.
The results are presented for $10^4$ Monte Carlo runs, SNR $\rho_0 = 10$~dB,  SINR threshold $\gamma_{\text{th}} \!=\! 10$~\cite{ref_srivatsa_chest_tsp_2022}.
We use the truncated Soliton distribution~\cite{ref_narayanan_istc_2012} $\phi(x) \!=\! 0.625x^2 + 0.25 x^3 + 0.125 x^4$ to generate the repetition factors of the users~\cite{ref_srivatsa_chest_tsp_2022}.\footnote{We do not optimize the repetition distribution in this work since the goal is to evaluate the impact of censoring.} 
The repetition factor $d_{i}$ is used to form the access vector for the $i$th user, by uniformly randomly choosing $d_{i}$ slots from $T$ slots without replacement~\cite{ref_liva_toc_2011}.
The packet replicas are transmitted in these $d_{i}$ slots.
For the empirical results, the number of users $M$ is computed based on $L$ as $M \!=\! \lfloor L T \rceil$; whereas the theoretical performance is dependent only on $L$, as described in Sec.~\ref{sec_de}.

\begin{figure}[t]
	\vspace{-0.3cm}
	\centering
	\includegraphics[width=0.5\textwidth]{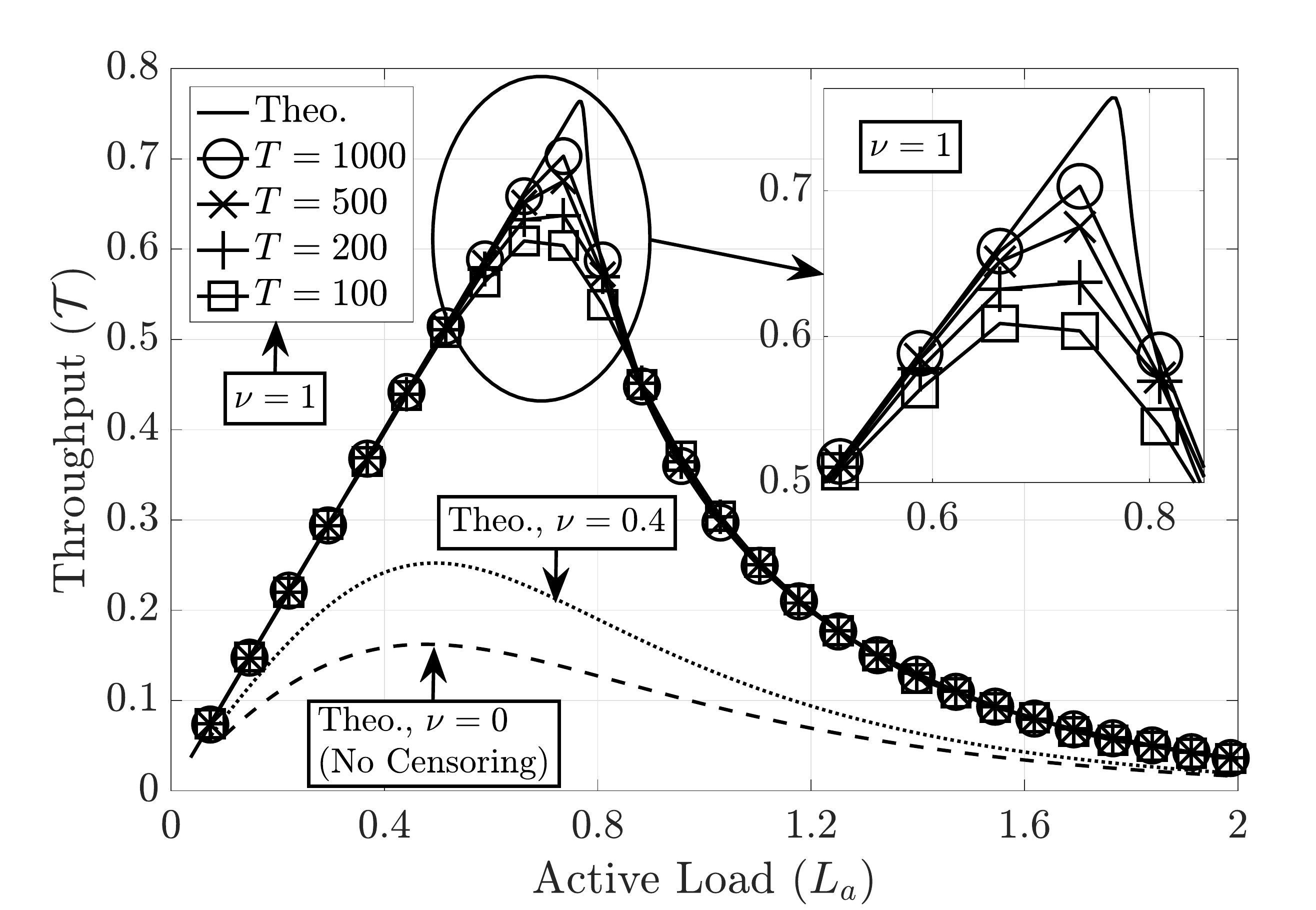}
	\caption{Impact of $T$ on the throughput.}
\label{fig_userc_de_verif}
	\vspace{-0.2cm}
\end{figure}
\begin{figure}[t]
	\vspace{-0.3cm}
	\centering
	\includegraphics[width=0.5\textwidth]{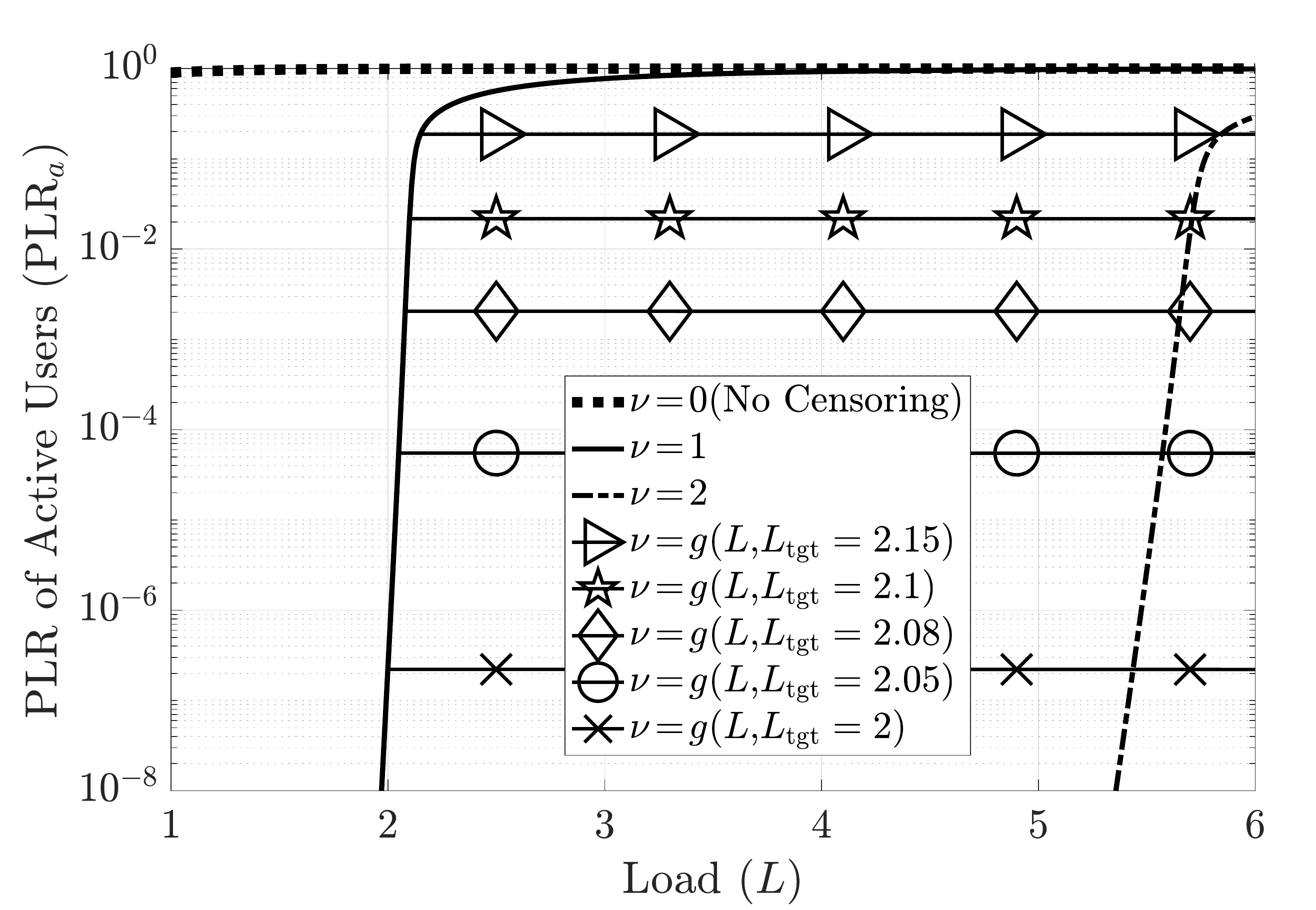}
	\caption{Choice of target load $L_{\text{tgt}}$ using theoretical ${\sf{PLR}}_a$.}
\label{fig_userc_plr_choice_L_tgt}
	\vspace{-0.2cm}
\end{figure}

Fig. \ref{fig_userc_de_verif} shows the impact of $T$ on the empirical throughput with $\nu \!=\! \rho_0^{-1}\gamma_{\text{th}} \!=\! 1$.
The theoretical asymptotic throughput curves for $\nu \!=\! 0, 0.4,$ and $1$, obtained via DE, are also shown.
The curves linearly increase till a peak, after which they drop quickly to zero as the system becomes MUI-limited.
The asymptotic $\mathcal{T}$ is maximized at $L_a^* \!=\! \mathcal{T}\! =\! 0.76$, for $\nu \!=\! 1$.
The linear increase in $\mathcal{T}$ marks the region in which ${\sf{PLR}}_a \!=\! 0$~\cite{ref_liva_toc_2011}: 
when $L_a \!\leq\! 0.76$, all active users are decoded.
Conventional IRSA corresponds to no censoring ($\nu\!=\!0$).
At $L_a\!=\!0.4$,  $\nu\!=\!0$ achieves $\mathcal{T} \!= \!0.15$, whereas $\nu\!=\!1$ achieves full throughput of $\mathcal{T} \!=\! L_a\!=\!0.4$.
The asymptotic throughput dramatically improves as $\nu$ is increased from $0$ to $1$, because users with poor channel states are self-censored.
Even with a little amount of censoring, C-IRSA performs better than IRSA.
Thus, C-IRSA helps overcome packet losses due to both poor CSI and MUI.

We have seen that the choice of the threshold $\nu$ must be such that $L_a \!\leq\! L_a^*$. 
In Fig.~\ref{fig_userc_plr_choice_L_tgt}, we depict the influence of the choice of the target load, $L_{\text{tgt}}$, using the asymptotic active PLR, ${\sf{PLR}}_a$.
The PLR is close to 1 with no censoring.
The ${\sf{PLR}}_a$ of the system increases with $L$ for $\nu\!=\!0,1,$ and $2$, and becomes $1$ at high loads.
The curves with $\nu\!=\!g(L,L_{\text{tgt}})$ follow the performance of $\nu \!= \!\rho_0^{-1} \gamma_{\text{th}} \!=\! 1$ upto a load of $L \!=\! L_{\text{tgt}}$, and beyond that ${\sf{PLR}}_a$ stays constant at every load.
Fixing a ${\sf{PLR}}_{a, \text{tgt}}$ yields the choice of $L_{\text{tgt}}$ and the corresponding threshold $\nu\!=\!g(L,L_{\text{tgt}})$.
The asymptotic PLR increases very quickly around the inflection load $L^*$.
In practice, however, choosing $L_{\text{tgt}} \!=\! 0.9L^*$ or $0.8L^*$ works well.

\begin{figure}[t]
	\vspace{-0.3cm}
	\centering
	\includegraphics[width=0.5\textwidth]{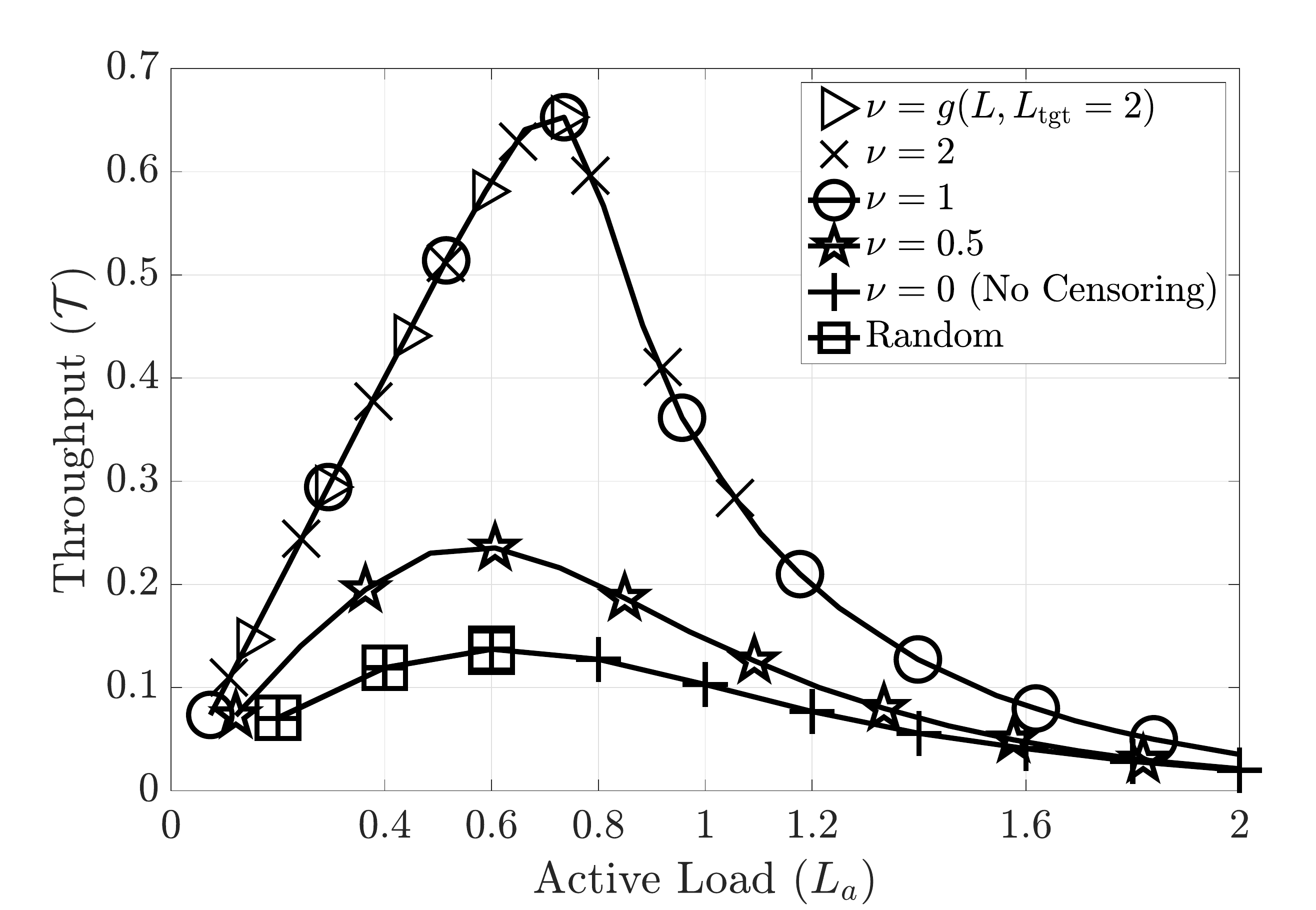}
	\caption{Effect of active load $L_a$ on $\mathcal{T}$.}
\label{fig_userc_thpt_vs_activeL}
	\vspace{-0.2cm}
\end{figure}

In Fig.~\ref{fig_userc_thpt_vs_activeL}, we show the effect of the active load $L_a$ on the empirical throughput $\mathcal{T}$, with $T\!=\!250$.
Conventional IRSA (no censoring, i.e., $\nu\!=\!0$) achieves very low throughputs since the system is highly interference limited.
Similar to the previous plot, where the theoretical throughput increased with increase in $\nu$, the empirical throughput also increases with an increase in from $\nu \!= \!0$ to $\nu \!=\! \gamma_{\text{th}}/\rho_0 \!=\! 1$.
For $\nu \!\geq \!\gamma_{\text{th}}/\rho_0,$ 
the throughput of the system stays constant with respect to the active load and the system achieves the same throughput for $\nu\!=\!2$ as for $\nu\!=\!1$.
From the plot,  we also see that we should choose a threshold $\nu$ such that we always operate the system at active load of $L_a \!\leq \!L_a^* \!= \!0.65$.
Also, by optimally choosing the threshold using $\nu\!=\!g(L,L_{\text{tgt}})$  as described in Sec.~\ref{sec_choice_nu}, we can obtain the same throughput as that obtained with $\nu\!=\!1$.
Note that in the MUI-limited regime, the PLR of IRSA is nonzero, and both users with poor channel states as well as users who collide with many users cannot be decoded correctly.
Censoring improves the performance of the system on both counts by choosing users whose packets are more likely to be decoded correctly as well as reducing the number of collisions.

\begin{figure}[t]
	\vspace{-0.3cm}
	\centering
	\includegraphics[width=0.5\textwidth]{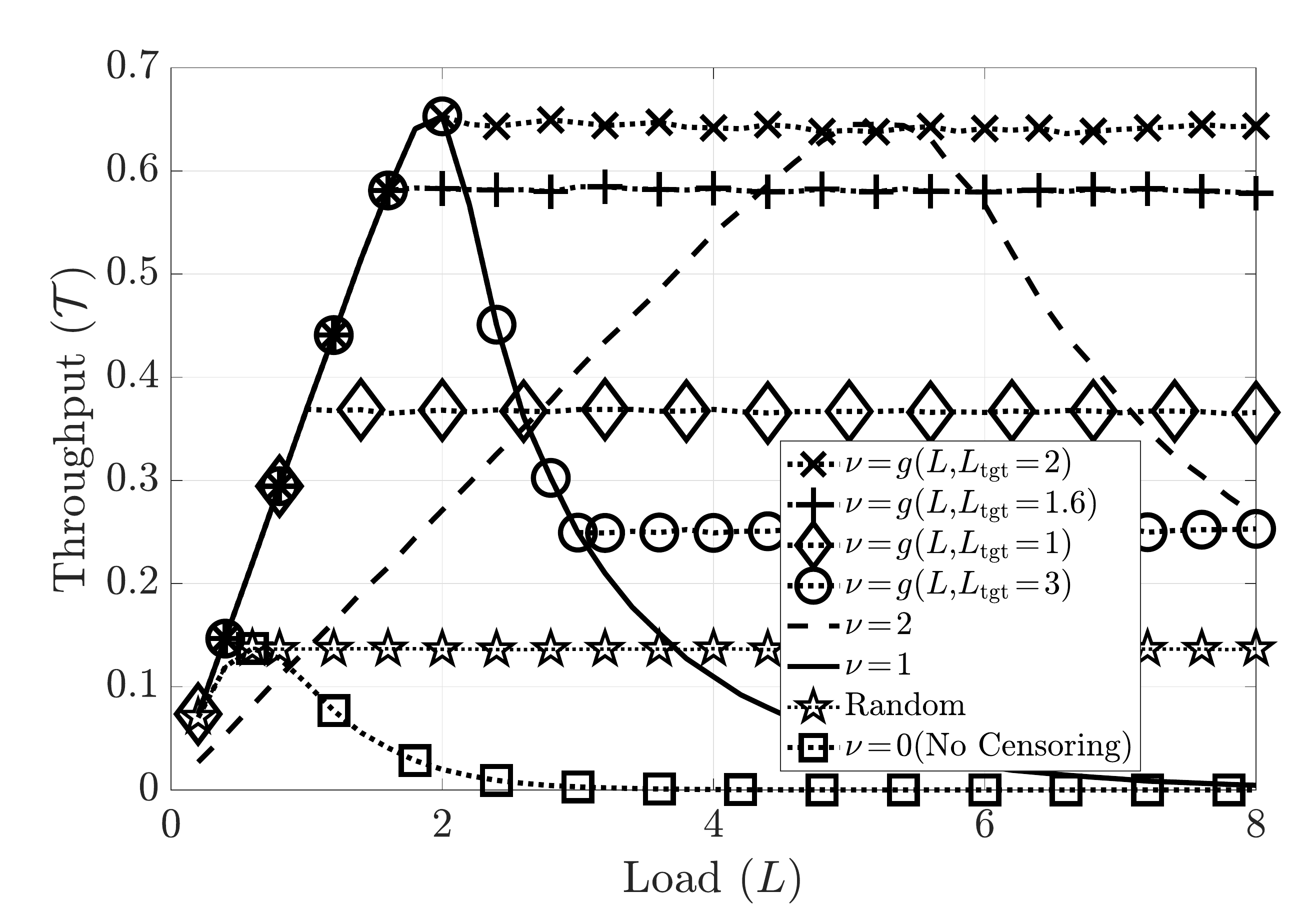}
	\caption{Impact of threshold $\nu$ on $\mathcal{T}$.}
\label{fig_userc_thpt_vs_L}
	\vspace{-0.2cm}
\end{figure}

So far, we have observed that both the theoretical and empirical throughputs are maximized at $\nu \!=\! \rho_0^{-1}\gamma_{\text{th}}$ for every $L_a$.
We now study the effect of censoring and the system load $L$ on the empirical throughput in Fig.~\ref{fig_userc_thpt_vs_L}, with $T\!=\!250$.
With $\nu\!=\!0$, i.e., no censoring, the throughput of IRSA becomes zero at $L\!=\!3$.
With $\nu\!=\!\rho_0^{-1}\gamma_{\text{th}}\!=\!1$, the throughput of the system increases linearly with load upto $\mathcal{T}\!=\!0.65$ at $L\!=\!2$, and beyond that, the throughput drops to zero.
This is also observed with $\nu\!=\!2$, which achieves a peak throughput of $\mathcal{T}\!=\!0.65$ at $L\!=\!5$.
The linearity of the curve upto $L\!=\!5$ indicates that too many users are self-censoring, and we could reduce $\nu$.
For $\nu\!=\!1$, we have ${\sf{PLR}}_a \!=\! 0$ and ${\sf{PLR}} \!=\! \bar{\rm F}(1)$ upto $L\!=\!2$; for $\nu\!=\!2$, we have ${\sf{PLR}}_a \!=\! 0$ and ${\sf{PLR}} \!=\! \bar{\rm F}(2)$ upto $L\!=\!5$.
Thus, we could choose $\nu$ for every $L$ such that we obtain an envelope of all curves for $\nu\!\geq\!1$, which yields the same performance as that of the curve marked $\nu\!=\!g(L,L_{\text{tgt}}\!=\!2)$.\footnote{The theoretical throughputs for Figs.~\ref{fig_userc_thpt_vs_activeL} and \ref{fig_userc_thpt_vs_L} match the above observations. Due to lack of space, we have not included them.  Also, the results are presented for $\rho_0^{-1}\gamma_{\text{th}} \!=\! 1$. The trends are similar for any other $\rho_0^{-1}\gamma_{\text{th}}$, and $\mathcal{T}$ is maximized at $\nu \!=\! \rho_0^{-1}\gamma_{\text{th}}$ for every $L_a$.}
All the curves marked $\nu\!=\!g(L,L_{\text{tgt}})$ follow the performance of $\nu\!=\!1$ upto $L_{\text{tgt}}$, beyond which $\mathcal{T}$ stays constant for every $L$.
Since $L^* \!=\! 2,$ choosing $L_{\text{tgt}} \!=\! 3$ is not preferred since the system is operating at a high PLR.
Choosing $L_{\text{tgt}} \!=\! 1,1.6,$ and $2$ all yield ${\sf{PLR}}_a \!=\! 0$ at all $L$ since the active load $L_a\!\leq\! 0.65$.
We thus choose $L_{\text{tgt}} \!=\!2$ to maximize $\mathcal{T}$, which can be obtained from our analysis as $L_{\text{tgt}} \!=\! L_a^*/\bar{\rm F}(\rho_0^{-1}\gamma_{\text{th}})  \!= \!0.65/\bar{\rm F}(1) \! =\! 2$.
Since the DE curves are achieved asymptotically,  in practice, we back off from $L_{\text{tgt}}$ by 10\% to 20\% to $L_{\text{tgt}} \!=\! 1.8$ or $1.6$, to achieve zero ${\sf{PLR}}_a$ at all $L$.
At high $L$, we see that C-IRSA with $L_{\text{tgt}}\! \leq \!L^*$ operates with $\mathcal{T}\!=\!0.65$, whereas conventional IRSA has $\mathcal{T}\!=\!0$.
Thus, the system can be operated at its maximum potential in C-IRSA compared to vanilla IRSA which has zero throughput.

\subsubsection{Impact of random censoring}
The censoring of users can be done in a random fashion as opposed to CSI-based censoring: users independently participate in each frame with probability $p_a$, and self-censor with probability $1-p_a$. 
This yields an active load of $L_a \!=\! L p_a$. The optimal random censoring can be done by choosing $p_a \!=\! L_a^*/L$, since this ensures that $L_a \!=\! L_a^*$.
The curve marked ``Random'' in Fig.~\ref{fig_userc_thpt_vs_activeL} uses random censoring and achieves the same throughput as conventional IRSA for every $p_a \!\in\! (0,1]$.
For the same active load $L_a$, the channel states of the uncensored users with CSI-based censoring are better than the channel states of the active users with random censoring.
With optimal random censoring, in order to operate the system at $\mathcal{T}_{\text{tgt}} \!=\! 0.15$ at $L_a^* \!=\! 0.6$, we need to choose $p_a \!=\! \min \{1,0.6/L\}$.
With this choice of $p_a$,  we obtain the curve marked ``Random'' in Fig.~\ref{fig_userc_thpt_vs_L}, which achieves $\mathcal{T} \!=\! 0.15$ at all $L\!\geq\!0.6$.
Thus, the optimal CSI-based censoring in C-IRSA achieves a peak throughput of $\mathcal{T} \!=\! 0.65$ whereas optimal random censoring in IRSA has a peak throughput of $\mathcal{T} \!=\! 0.15$,  an over $4\times$ improvement.

\section{Conclusion} \label{sec_conc}
In this work, we proposed a variant of IRSA, called C-IRSA, to overcome the performance degradation of IRSA at high loads. In C-IRSA, users self-censor based on their CSI, and the protocol retains the fully distributed, random access nature of IRSA. 
We derived closed-form expressions for the success probability $\theta_{r}$ for CSI-based censoring, and theoretically characterized the asymptotic performance of C-IRSA.
Our analysis allows us to determine the optimal choice of the censor threshold $\nu$, with which the PLR of the active users can be driven close to zero and yields the highest possible throughput.
The results showed that we can achieve a $4\times$ improvement in C-IRSA compared to optimal random censoring.
Future work could account for CSI and load estimation errors, optimize the repetition distribution under C-IRSA, and also include a proportional fairness mechanism for users with poor CSI.

\appendices
\renewcommand{\thesectiondis}[2]{\Alph{section}:}
\section{Proof of Theorem \ref{thm_thetar}} \label{appendix_theta_r}
We now characterize $\theta_{r}$, which is the probability of decoding a reference packet in a \emph{single slot} where $r$ users have transmitted their packets.
Since there is only one slot under consideration,  users are decoded via intra-slot SIC.
The reference packet is one of the $r$ packets, and it gets decoded only if the packets having a higher SINR get successfully decoded first.
Hence, they must also satisfy the SINR $\geq \gamma_{\text{th}}$ constraint.
Thus, $\theta_{r}$ is the probability that the reference packet and the packets with higher SINRs all get decoded. 

We denote the set of active users who have not yet been decoded in the first $k-1$ intra-slot decoding iterations by $\mathcal{S}_{k}$, and $\mathcal{S}_k^m \triangleq \mathcal{S}_k \setminus \{m\}$, with $\mathcal{S}_1 \!=\! [r]$.
The SINR of the $m$th user in the $k$th intra-slot decoding iteration,  $\rho_{m}^k$, is calculated~as $\rho_{m}^k \!=\! \textstyle{ | h_{m} |^2/(\rho_0^{-1} \!+\!  \sum\nolimits_{i \in \mathcal{S}_k^m}^{r} | h_{i} |^2 )}.$
Let $\rho_{\max}^k$ denote the SINR of the user with the highest SINR in the $k$th intra-slot decoding iteration, calculated as $\rho_{\max}^k \!=\! \max_{m \in \mathcal{S}_k} \ \rho_{m}^k$.
Let~$s$ be the index of the intra-slot decoding iteration in which the reference packet is decoded, with $1 \!\leq\! s \!\leq\! r$.
Thus, $\theta_r$ is calculated as $\theta_{r} \!=\! \text{Pr} ( \rho_{\max}^1\! \geq\! \gamma_{\text{th}}, \rho_{\max}^2 \!\geq\! \gamma_{\text{th}},  \ldots,  \rho_{\max}^{s} \!\geq\! \gamma_{\text{th}}  ).$
Since the reference packet is tagged uniformly at random from the users, the reference packet is equally likely to get decoded in any decoding iteration.
We denote the probability that the $k$ packets with the highest SINRs across decoding iterations all exceed the threshold ${\gamma}_{\text{th}}$ by $\theta_{rk} \!\triangleq\! \text{Pr} ( \rho_{\max}^1 \!\geq\! \gamma_{\text{th}}, \rho_{\max}^2 \!\geq\! \gamma_{\text{th}},  \ldots,  \rho_{\max}^{k} \!\geq\! \gamma_{\text{th}}  )$.
We can calculate $\theta_r$ using $\theta_{rk}$ as  $\theta_{r} \!=\!  ( \sum\nolimits_{k=1}^r \theta_{rk})/r.$
Without loss of generality, let the channels of the users be ordered as $| h_{1} |^2 \geq | h_{2} |^2 \geq \ldots \geq | h_{r} |^2$.
Now, 
\begin{align}
\theta_{rk} &= \text{Pr} \left({\dfrac{  | h_{1} |^2}{  \rho_0^{-1} \! + \!  \sum\nolimits_{i=2}^{r} | h_{i} |^2  }\!  \geq \! \gamma_{\text{th}},  \dfrac{  | h_{2} |^2}{  \rho_0^{-1} \! +  \! \sum\nolimits_{i=3}^{r} | h_{i} |^2  } \! \geq \! \gamma_{\text{th}}, }
\right. \nonumber \\ 
& \left. 
\! \ldots,  \!  {\dfrac{  | h_{k} |^2}{  \rho_0^{-1} \! + \!  \sum\nolimits_{i=k+1}^{r} | h_{i} |^2  } \! \geq \! \gamma_{\text{th}} \bigg| |h_j|^2 \geq \nu,  \forall j \in [r] \!} \right) \! . \label{eqn_theta_rk11}
\end{align} 
The above is a conditional probability, conditioned on $|h_{j}|^2\geq \nu$,  since we are considering only uncensored users.
Thus,  $\theta_{rk}$ from \eqref{eqn_theta_rk11} can be calculated equivalently as
\begin{align*}
& \theta_{rk} = \text{Pr} ( t_{1}   \geq \gamma_{\text{th}} (\rho_0^{-1} + \textstyle{\sum\nolimits_{i=2}^{r}}  t_{i} ),   t_{2}   \geq \gamma_{\text{th}} (\rho_0^{-1} \! + \!  \textstyle{\sum\nolimits_{i=3}^{r}}  t_{i} ),  
\nonumber 
\\ 
& \qquad \qquad \qquad
 \ldots,   t_{k}   \geq \gamma_{\text{th}} (\rho_0^{-1} \! + \!  \textstyle{\sum\nolimits_{i=k+1}^{r}}  t_{i} )  ). \! \nonumber 
\end{align*}
Here, $t_{i} $ is a random variable follows a \emph{truncated exponential} distribution with the density function ${\rm f}(t) = \exp ( \nu - t) \cdot \mathbbm{1}\{ \nu \leq t < \infty \}.$
Assuming $\nu\leq \rho_0^{-1} \gamma_{{\text{\rm{th}}}}$,  with $\bar{\gamma}_{{\text{\rm{th}}}, i} = (1+ \gamma_{{\text{\rm{th}}}})^i$, $ \theta_{rk}$ can be calculated as
\begin{align}
& \theta_{rk} = e^{{r\nu}} \int_{\nu}^{\infty} \! \! e^{-t_{r}} {\rm d}t_{r} \int_{\nu}^{\infty} \! \! e^{-t_{r-1}} {\rm d}t_{r-1} \cdots \int_{\nu}^{\infty} \! \! e^{-t_{k+1}} {\rm d}t_{k+1} \nonumber 
\\
& \ \ \ \ \times  \int_{\gamma_{\text{th}} (\rho_0^{-1} \! + \!  {\sum\nolimits_{i=k+1}^{r}}  t_{i} )}^{\infty} \!\! e^{-t_{k}} {\rm d}t_{k}  \cdots   \int_{\gamma_{\text{th}} (\rho_0^{-1} \! + \!  {\sum\nolimits_{i=2}^{r}}  t_{i} )}^{\infty} \!\! e^{-t_{1}} {\rm d}t_{1} \nonumber 
\\
& \ \ =  \dfrac{ \exp ( r \nu - (r-k) \nu \bar{\gamma}_{{\text{th}}, k} -\rho_0^{-1} \gamma_{\text{th}} (\sum_{i=1}^{k} \bar{\gamma}_{{\text{\rm{th}}}, i-1} )  )}{\bar{\gamma}_{{\text{th}}, k}^{r - (k+1)/2}}.
\end{align} 
Thus, we get
\begin{equation} 
\theta_{r} \! = \! \textstyle{\sum\limits_{k=1}^r} \dfrac{ \exp ( r \nu - (r-k) \nu \bar{\gamma}_{{\text{\rm{th}}}, k} -\rho_0^{-1} (\bar{\gamma}_{{\text{\rm{th}}}, k} - 1)  )}{ r \ \bar{\gamma}_{{\text{\rm{th}}}, k}^{r - (k+1)/2}}.
\end{equation} 

\vspace{0.3cm}
\bibliographystyle{IEEEtran}
\bibliography{IEEEabrv,my_refs}

\begin{thebibliography}{10}
\providecommand{\url}[1]{#1}
\csname url@samestyle\endcsname
\providecommand{\newblock}{\relax}
\providecommand{\bibinfo}[2]{#2}
\providecommand{\BIBentrySTDinterwordspacing}{\spaceskip=0pt\relax}
\providecommand{\BIBentryALTinterwordstretchfactor}{4}
\providecommand{\BIBentryALTinterwordspacing}{\spaceskip=\fontdimen2\font plus
\BIBentryALTinterwordstretchfactor\fontdimen3\font minus
  \fontdimen4\font\relax}
\providecommand{\BIBforeignlanguage}[2]{{%
\expandafter\ifx\csname l@#1\endcsname\relax
\typeout{** WARNING: IEEEtran.bst: No hyphenation pattern has been}%
\typeout{** loaded for the language `#1'. Using the pattern for}%
\typeout{** the default language instead.}%
\else
\language=\csname l@#1\endcsname
\fi
#2}}
\providecommand{\BIBdecl}{\relax}
\BIBdecl

\bibitem{ref_chen_jsac_2021}
X.~Chen, D.~W.~K. Ng, W.~Yu, E.~G. Larsson, N.~Al-Dhahir, and R.~Schober,
  ``Massive access for {5G} and beyond,'' \emph{{IEEE} J. Sel. Areas Commun.},
  vol.~39, no.~3, pp. 615--637, 2021.

\bibitem{ref_liva_toc_2011}
G.~{Liva}, ``Graph-based analysis and optimization of contention resolution
  diversity slotted {ALOHA},'' \emph{{IEEE} Trans. Commun.}, vol.~59, no.~2,
  pp. 477--487, February 2011.

\bibitem{ref_srivatsa_chest_tsp_2022}
C.~R. Srivatsa and C.~R. Murthy, ``On the impact of channel estimation on the
  design and analysis of {IRSA} based systems,'' \emph{{IEEE} Trans. Signal
  Process.}, vol.~70, pp. 4186--4200, 2022.

\bibitem{ref_srivatsa_uad_tsp_2022}
------, ``User activity detection for irregular repetition slotted aloha based
  {MMTC},'' \emph{{IEEE} Trans. Signal Process.}, vol.~70, pp. 3616--3631,
  2022.

\bibitem{ref_paolini_tit_2015}
E.~{Paolini}, G.~{Liva}, and M.~{Chiani}, ``Coded slotted {ALOHA}: A
  graph-based method for uncoordinated multiple access,'' \emph{{IEEE} Trans.
  Inf. Theory}, vol.~61, no.~12, pp. 6815--6832, Dec 2015.

\bibitem{ref_khaleghi_pimrc_2017}
E.~E. {Khaleghi}, C.~{Adjih}, A.~{Alloum}, and P.~{Muhlethaler}, ``Near-far
  effect on coded slotted {ALOHA},'' in \emph{Proc. PIMRC}, Oct 2017.

\bibitem{ref_clazzer_icc_2017}
F.~{Clazzer}, E.~{Paolini}, I.~{Mambelli}, and C.~{Stefanovic}, ``Irregular
  repetition slotted {ALOHA} over the {R}ayleigh block fading channel with
  capture,'' in \emph{Proc. ICC}, May 2017.

\bibitem{ref_srivatsa_spawc_2019}
C.~R. Srivatsa and C.~R. Murthy, ``Throughput analysis of {PDMA/IRSA} under
  practical channel estimation,'' in \emph{Proc. SPAWC}, July 2019.

\bibitem{ref_srivatsa_spawc_2022}
------, ``Performance analysis of irregular repetition slotted aloha with
  multi-cell interference,'' in \emph{Proc. SPAWC}, July 2022.

\bibitem{ref_choi_wcl_2021}
J.~Choi and J.~Ding, ``Network coding for {$K$}-repetition in grant-free random
  access,'' \emph{{IEEE} Wireless Commun. Lett.}, vol.~10, no.~11, 2021.

\bibitem{ref_ding_cscn_2021}
J.~Ding and J.~Choi, ``{SIC} aided {$K$}-repetition for mission-critical {MTC}
  in cell-free massive {MIMO},'' in \emph{Proc. CSCN}, 2021.

\bibitem{ref_narayanan_istc_2012}
K.~R. {Narayanan} and H.~D. {Pfister}, ``Iterative collision resolution for
  slotted {ALOHA}: An optimal uncoordinated transmission policy,'' in
  \emph{Proc. ISTC}, Aug 2012, pp. 136--139.

\end{thebibliography}

\end{document}